\documentclass[letterpaper, 10 pt, conference]{ieeeconf}  

\IEEEoverridecommandlockouts                      

\usepackage{amsmath,graphicx,psfrag,epsfig,color,amsfonts,hyperref}
\usepackage{version}
\usepackage{subfigure}
\usepackage{breqn}
\usepackage{mathtools,amssymb,lipsum,mathrsfs}

\usepackage{cuted}
\usepackage{bibentry}

\def\expt{\mathbb{E}}
\def\real{\mathbb{R}}
\def\comp{\mathbb{C}}

\newcommand\oprocendsymbol{\hbox{$\blacksquare$}}
\newcommand\oprocend{\relax\ifmmode\else\unskip\hfill\fi\oprocendsymbol}

\renewcommand{\theenumi}{\roman{enumi}}

\newtheorem{proposition}{Proposition}
\newtheorem{theorem}{Theorem}
\newtheorem{definition}{Definition}

\newtheorem{corollary}{Corollary}

\newtheorem{example}{Example}
\newtheorem{notation}{Notation}

\title{Stochastic Logarithmic Lipschitz Constants:\\A Tool to Analyze  Contractivity of Stochastic Differential Equations}

\author{Zahra Aminzare$^{1}$ 
\thanks{$^{1}$Zahra Aminzare is with the Department of Mathematics,  University of Iowa, IA, USA.
	{\tt\small zahra-aminzare@uiowa.edu}}
}

\def \mc {\mathcal}
\def \mcfg{\mathcal{M}_{F,G}^{(h,W)}}

\def \mcFG{\mathcal{M}_{F,G }^{(h,W)}}
\begin{document}

\maketitle
\thispagestyle{empty}
\pagestyle{empty}

\begin{abstract}
 We introduce the notion of stochastic logarithmic Lipschitz constants and use these constants to characterize stochastic contractivity of Itô stochastic differential equations (SDEs) with multiplicative noise. 
We find an upper bound for stochastic logarithmic Lipschitz constants
based on known logarithmic norms (matrix measures) of the Jacobian of the drift and diffusion terms of the SDEs. 
We discuss noise-induced contractivity in SDEs and common noise-induced synchronization in network of SDEs and illustrate the theoretical results on a noisy Van der Pol oscillator. We show that a deterministic Van der Pol oscillator is not contractive. But, adding a multiplicative noise makes the underlying SDE stochastically contractive. 
\end{abstract}

\section{Introduction}\label{sec:intro}

\noindent Contraction theory is a methodology for assessing the global convergence of trajectories of a dynamical system to each other instead of convergence to a pre-specified attractor. 
Given a vector norm 
with its induced matrix norm, the \textit{Logarithmic norm} of a linear operator $A$  is defined as the directional derivative of the matrix norm in the direction of $A$ and evaluated at the identity matrix, \cite{Lewis1949,
Hartman1961}. This definition can be extended to any nonlinear operator using the notion of \textit{logarithmic Lipschitz constant} \cite{Soderlind}. Logarithmic Lipschitz constants of a nonlinear vector field or the logarithmic norm of the Jacobian of the vector field can characterize the contraction property of a nonlinear system. 

 Studying contractivity of ordinary differential equations (ODEs) and reaction-diffusion partial differential equations using logarithmic norms and logarithmic Lipschitz constants is a well-established research topic (see e.g. \cite{dahlquist, Demidovich1961, Demidovich1967, Yoshizawa1966, 
Yoshizawa1975, Soderlind2, Loh_Slo_98, arcak2011_Automatica, Simpson-Porco_Bullo_2014, russo_dibernardo_sontag13, contraction_survey_CDC14, Coogan_Arcak_2015, Margaliot_Sontag_Tuller, 2016_Aminzare_Sontag, Coogan_2019, CV_J_Bullo_2020, D_J_Bullo_2021}). However, there are not too many attempts to study contractivity of non-deterministic systems and in particular, Itô stochastic differential equations (SDEs). 
In \cite{slotine_stochastic,Han_Chung_2021} contraction theory is studied using stochastic Lyapunov function and incremental stability. 
 In \cite{2013_Tabareau_Slotine,LG-TM-MM-2021,newman_2018} contraction theory is studied for random dynamical systems.
 In \cite{2013_Pham_Slotine,2019_Bouvrieand_Slotine} contractivity is generalized to Riemannian metrics and Wasserstein norms, respectively.
 In \cite{Dani_Chung_Hutchinson_2015}, stochastic contraction is studied for Poisson shot noise and finite-measure Lévy noise. 
This work takes a step forward and extends contraction theory to SDEs using generalized forms of logarithmic norms and logarithmic Lipschitz constants. 

Stochastic contraction theory can be used to study the stability of SDEs and to characterize the synchronization behavior of networks of nonlinear and noisy systems.  Synchronization induced by common noise has been observed experimentally  
and confirmed theoretically  
in many networks of nonlinear dynamical systems without mutual coupling. Indeed, this kind of synchronization is equivalent to the stochastic contraction of SDEs that we study in Section \ref{sec:noise-induced:contrac} below.  Therefore, extending contraction theory to SDEs can be beneficial for characterizing networks' synchronization. 

In \cite{Ahmad_Raha_2010}, the authors introduced stochastic logarithmic norms and used them to study the stability properties of linear SDEs.  Analog to the deterministic version, stochastic logarithmic norms are proper tools for characterizing the contractivity of linear SDEs, but they are not directly applicable to nonlinear SDEs. Our main contributions are to generalize the notion of logarithmic Lipschitz constants, which generalize the logarithmic norms to nonlinear operators, and use them to study the contractivity of nonlinear SDEs. 

 The remainder of the paper is organized as follows. 
Section~\ref{sec:background} reviews logarithmic Lipschitz constants of deterministic nonlinear operators and  contraction properties of ODEs. 
Sections~\ref{sec:main-definitions} and ~\ref{sec:main-results} contain the  definition of stochastic logarithmic Lipschitz constants and main results on characterizing the stochastic contractivity of SDEs.
Section~\ref{sec:noise-induced:contrac} discusses how noise can induce stochastic contractivity and synchronization and illustrates the results in a numerical example. 
 Section~\ref{sec:conclusion} is the conclusion and discussion. 
Some of the proofs are given in an Appendix, Section~\ref{sec:appendix}.

\section{Background}\label{sec:background}
 In this section we review the definitions of logarithmic norms and logarithmic Lipschitz constants and explain how they are helpful to study contraction properties of ODEs.

\begin{definition}\textit{(\textbf{Logarithmic norm})} Let $(\mc X,\|\cdot\|_{\mc X})$ be a finite dimensional normed vector space over $\real$ or $\comp$. The space $\mathcal{L}(\mc X,\mc X)$ of linear transformations $A\colon \mc X \to \mc X$ is also a normed vector space with the induced operator norm $\|A\|_{\mc X\to \mc X}=\sup_{\|x\|_{\mc X}=1}\|Ax\|_{\mc X}.$ The \textit{logarithmic norm} (or matrix measure) of $A$ induced by $\|\cdot\|_{\mc X}$ is defined as the directional derivative of the matrix norm,
\begin{equation*}
\mu_{\mc X}[A]\;=\;\displaystyle\lim_{h\to 0^+}\frac{1}{h}\left(\|I+hA\|_{\mc X\to \mc X}-1\right),
\end{equation*}
where $I$ is the identity operator on $\mc X$. 
\end{definition}

\begin{definition}\label{def:LLC}\textit{(\cite{Soderlind}, \textbf{Logarithmic Lipschitz constants})} Assume $F\colon \mc Y\subseteq \mc X \to \mc X$ is an arbitrary function. Two generalizations of the logarithmic norms are 
the strong least upper bound (s-lub) and least upper bound (lub) \textit{logarithmic Lipschitz constants}, which are respectively defined by 

{\small{\begin{align*}\label{defM+}
&M_{\mc X}^{+}[F]=\displaystyle\sup_{u\neq v\in \mc Y}\lim_{h\to0^{+}}\frac{1}{h}\left(\frac{\|u-v+h(F(u)-F(v))\|_{\mc X}}{\|u-v\|_{\mc X}}-1\right),\\
&M_{\mc X}[F]=\displaystyle\lim_{h\to0^{+}}\sup_{u\neq v\in \mc Y}\frac{1}{h}\left(\frac{\|u-v+h(F(u)-F(v))\|_{\mc X}}{\|u-v\|_{\mc X}}-1\right).
\end{align*}}}
\end{definition}

\begin{proposition}\textit{(\cite{Soderlind, aminzare-sontag}, \textbf{Some properties of logarithmic Lipschitz constants})} 
 \noindent $M_{\mc X}^+$ and $M_{\mc X}$ are sub-linear, i.e.,  for $F, F^i\colon \mc Y\to \mc X,$ 
 and $\alpha\geq0$ 
 (similar properties hold for $M_{\mc X}$): 
 \begin{itemize}
\item $M_{\mc X}^+[F^1+F^2]\leq M_{\mc X}^+[F^1]+M_{\mc X}^+[F^2]$,  
 \item$M_{\mc X}^{+}[\alpha F]=\alpha M_{\mc X}^{+}[F],$ and 
 \item $M_{\mc X}^+[F]\leq M_{\mc X}[F].$
 \end{itemize}
\end{proposition}

\begin{proposition}\label{prop:logarithmic:constants:norm}\textit{(\cite{Soderlind2}, \textbf{Relationship between logarithmic Lipschitz constants and logarithmic norms})} 
For finite dimensional space $\mc X$, the (lub) logarithmic Lipschitz constant $M_{\mc X}$ generalizes the  logarithmic norm $\mu_{\mc X}$, i.e., for any matrix $A$, $M_{\mc X} [A]=\mu_{\mc X}[A].$ Furthermore, by the definitions, 
$M_{\mc X}[A] = M^+_{\mc X}[A] =\mu_{\mc X}[A].$
 Let $\mc Y$ be a connected subset of $\mc X$. For a globally Lipschitz and continuously differentiable function $F\colon \mc Y\to \real^n$,
$
\sup_{x\in \mc Y}\mu_{\mc X}[J_F(x)]\leq M_{\mc X}[F],
$
 where $J_F$ denotes the Jacobian of $F$. Moreover, if $\mc Y$ is a convex subset of $\mc X$, then
\[
\displaystyle\sup_{x\in \mc Y}\mu_{\mc X}[J_F(x)]= M_{\mc X}[F].
\]
 \end{proposition}

\begin{definition} \textit{(\textbf{Contractive ODE})} Consider 
\begin{align}\label{ODE}
\dot x =F(x,t),
\end{align}
where $x\in \mc Y\subset \real^n$ is an $n-$dim vector describing the state of the system,  $t\in[0,\infty)$ is the time, and $F$ is an $n-$dim nonlinear vector field. Assume that $\mc Y$ is convex and $F$ is continuously differentiable on $x$ and continuous on $t$. 
The system \eqref{ODE} is called contractive if there exists $c>0$ such that for any two solutions $X$ and $Y$ that remain in $\mc Y$, and $t\geq0$,
$\|X(t)-Y(t)\| \leq e^{-ct}\|X(0)-Y(0)\|.$
\end{definition}

\noindent In the following theorem and corollary, we find a value for the \textit{contraction rate} $c$ using logarithmic Lipschitz constant of the vector field $F$ and the logarithmic norm of the Jacobian of $F$ induced by a norm $\|\cdot\|_{\mc X}$ on $\real^n$.

\begin{theorem}\textit{(\cite[Proposition 3]{contraction_survey_CDC14}, \textbf{Contractivity of ODEs using logarithmic Lipschitz constants})}\label{thm:contractivity}
For a given norm $\|\cdot\|_{\mc X}$ and any two trajectories $X(t)$ and $Y(t)$ of \eqref{ODE}  and any $t\geq0$ the following inequality holds
\[\|X(t)-Y(t)\|_{\mc X} \leq e^{M_{\mc X}^+[F] t}\|X(0)-Y(0)\|_{\mc X}.\] 
In particular, if $M_{\mc X}^+[F]<0$, then \eqref{ODE}  is contractive. 
\end{theorem}

\begin{corollary} \textit{(\textbf{Contractivity of ODEs using logarithmic norms})}
Under the conditions of Theorem \ref{thm:contractivity}, if  $\mc Y$ is connected and  
$\sup_{(x,t)}\mu_{\mc X}[J_F(x,t)] \leq -c$, for some constant $c>0$ and some norm $\|\cdot\|_{\mc X}$, then \eqref{ODE} is contractive. 
\end{corollary}

\section{Stochastic Logarithmic Lipschitz Constants}\label{sec:main-definitions}

 In this section we generalize the definition of logarithmic Lipschitz constants given in Definition~\ref{def:LLC}. 

\begin{notation}\label{notation}
 We will use the following notations for the rest of the paper. 
\begin{itemize}
\item  $(\mc X,\|\cdot\|_{\mc X})$ is a normed space over $\real^n$ and $ \mc Y\subseteq \mc X$. 
\item $F\colon \mc Y\to \real^n$ is a vector field with components $F_i$. 
\item  $G$ is an $n\times d$ matrix of continuously differentiable  column vectors $G_{j}\colon \mc Y \to \real^n$, for $j=1,\ldots,d$.  
\item $W(t)$ is a $d-$dim Wiener process with components $W_j$.
\item  $\Delta W_j :=W_j(t+h)-W_j(t)= \int_t^{t+h} dW_j(s)$ and $\Delta W = (\Delta W_1, \ldots, \Delta W_d)^\top$
 \item $\Delta W^2_{i,j} := \int_t^{t+h} dW_i(s) \int_t^s dW_j(s')$ and $\Delta W^2$ is a $d\times d$ matrix with components $\Delta W^2_{i,j}$.
 \item For $k=1,\ldots,d$,  $L_k := \sum_{l=1}^n G_{lk}\frac{\partial}{\partial x_l}$.
 \item  $\mcFG$ is an $n-$dim function on $\mc Y$ with components:
\begin{align}\label{mcFG}
hF_i+ \sum_{j=1}^dG_{ij} \Delta W_j+ \sum_{j,k=1}^d L_k G_{ij} \Delta W^2_{j,k}.
\end{align}
\end{itemize}
\end{notation}
Let $\Delta_{jk}= \Delta W_j\Delta W_k+\Delta W^2_{j,k} -\Delta W^2_{k,j}$. Using \cite[Equation~(15.5.26)]{CG:09}, 
$\Delta W^2_{j,k} = \frac{1}{2} (\Delta_{jk}-h\delta_{jk}),$ where  $\delta_{jk}$ is the Kronecker delta. A straightforward calculation shows that
\[\mcFG = h\Big(F - \frac{1}{2}\sum_{i=1}^dJ_{G_j} G_j\Big)+\sum_{i=1}^d G_j\Delta W_j + \mathcal{R}, \]
where  $\mathcal{R} = \frac{1}{2}\sum_{j,k=1}^d L_k G_{ij} \Delta_{jk}$.  

\begin{definition}\textit{(\textbf{Stochastic logarithmic Lipschitz constants})} 
The s-lub and lub \textit{stochastic logarithmic Lipschitz constants} of $F$ and $G$ in the $l$-th mean and induced by $\|\cdot\|_{\mc X}$ are respectively:
{\small{\begin{align*}
&\mc M^{+}_{\mc X,l}[F,G]=\displaystyle\sup_{t}\sup_{u\neq v\in \mc Y}\lim_{h\to0^{+}}\frac{1}{h}\times\\
&\Big(\expt\;\frac{\|u-v+\mcFG(u) -\mcFG (v)\|^l_{\mc X}}{\|u-v\|^l_{\mc X}}-1\Big),\\
&\mc M_{\mc X,l}[F,G]=\displaystyle\sup_{t}\lim_{h\to0^{+}}\sup_{u\neq v\in \mc Y}\frac{1}{h}\times\\
&\Big(\expt\;\frac{\|u-v+\mcFG (u) -\mcFG (v)\|^l_{\mc X}}{\|u-v\|^l_{\mc X}}-1\Big),
\end{align*}}}
where $\expt$ denotes the expected value
\end{definition}

 In \cite{Ahmad_Raha_2010} the authors introduced the notion of stochastic logarithmic norm which is a special case of $\mc M_{\mc X, l}[F,G]$ with linear $F$ and $G_j$, i.e., $F(u) = Au$ and $G_j(u) = B_ju$ for square matrices $A, B_j$s. 

\begin{proposition}  \textit{(\textbf{Some properties of stochastic logarithmic Lipschitz constants})} \label{prop:properties:stochastic:constants}
Let $\alpha>0$ be a constant, $F, F^1$, and $F^2$ be vector functions as described in Notation \ref{notation}, and $G, G^1,$ and $G^2$ be matrices as described in Notation \ref{notation}. The following statements hold.

\noindent 1.  For a zero matrix $G$, $\mc M_{\mc X, l}^{+}[F,0] = l M_{\mc X}^{+}[F]$. 

\noindent 2. $\mc M_{\mc X, l}^{+}[F,G] \leq \mc M_{\mc X, l}[F,G]. $

\noindent 3. Unlike the deterministic ones, the stochastic  logarithmic Lipschitz constants are not sub-linear. However, they satisfy:
  \begin{itemize}
  \item $\mc M_{\mc X, l}^{+}[\alpha F,\sqrt\alpha G] =\alpha \mc M_{\mc X, l}^{+}[ F, G], \text{and}$
 \item $\mc M_{\mc X, l}^{+}[ F^1+F^2, G^1+G^2] \\\leq  \mc M_{\mc X, l}^{+}\Big[ F^1, \frac{G^1+G^2}{\sqrt 2}\Big] +  \mc M_{\mc X, l}^{+}\Big[ F^2, \frac{G^1+G^2}{\sqrt 2}\Big].$
\end{itemize}
  Similar properties hold for $\mc M_{\mc X, l}$.
\end{proposition}

\noindent A proof is given in Appendix, Section \ref{sec:appendix}. 
%
\section{Contraction Properties of SDEs}\label{sec:main-results}

 In this section we first define stochastic contractivity and then provide conditions that guarantee contractivity  in SDEs. 
Consider 
\begin{align}\label{SDE:eq}
dX(t) = F(X(t))dt + G(X(t)) dW(t),
\end{align}
where all the terms are as defined in Notation \ref{notation}. Furthermore, we assume that $F$ and $G$ satisfy the Lipschitz and growth conditions: $\exists K_1,K_2>0$ such that $\forall x,y$:
  
$\|F(x)-F(y)\|+\|G(x)-G(y)\|\leq K_1\|x-y\|,$ and
 
$\|F(x)\|^2+\|G(x)\|^2\leq K_2(1+\|x\|^2)$, 

\noindent where $\|\cdot\|$ denotes the Euclidean norm.  Note that for the matrix $G, \|G\|^2= \sum_{i,j}|G_{ij}|^2$. Under these conditions, for any given initial condition $X(0)$ (with probability one) the SDE has a unique non-anticipating solution, i.e., the solution is independent of the Wiener process, see \cite[Chapter 4]{CG:09}. 

\begin{definition}[\textbf{Stochastic contraction}]\label{def:stochastic:contraction}
An SDE described by \eqref{SDE:eq}
 is $l-$th moment contractive if there exists a  constant $c>0$ such that for any solutions $X(t)$ and $Y(t)$ with initial conditions  $X(0)$ and $Y(0)$, and $\forall t\geq0$,
 \begin{equation}\label{eq:moment:stability}
\expt \|X(t) - Y(t)\|_{\mc X}^l \;\leq\; \expt \|X(0) - Y(0)\|_{\mc X}^l \;\; e^{-c t}.
\end{equation}
 \end{definition} 

\begin{theorem}\label{thm:stochastic:contractivity}\textbf{\textit{(Stochastic contraction based on stochastic logarithmic Lipschitz constants)}} 
For any two solutions $X(t)$ and $Y(t)$ of \eqref{SDE:eq} and $\forall t\geq0$,
 \begin{equation}\label{eq:contractivity:1}
\expt \|X(t) - Y(t)\|_{\mc X}^l \;\leq\; \expt \|X(0) - Y(0)\|_{\mc X}^l \;\; e^{\mc M^{+}_{\mc X,l}[F,G] t}.
\end{equation}
Moreover, if $\mc M^{+}_{\mc X,l}[F,G]\leq -c$ for some $c>0$, \eqref{SDE:eq} becomes  $l-$th moment stochastically contractive. 
\end{theorem}

\begin{proof}
 If $\expt \|X(t) - Y(t)\|_{\mc X}^l=0$, then \eqref{eq:contractivity:1} holds. Therefore, we assume that $\expt \|X(t) - Y(t)\|_{\mc X}^l\neq0$. 
Using  Milstein algorithm \cite[Chapter 15]{CG:09} any solution $X(t)$ can be approximated by 
\begin{equation*}
X(t+h) = X(t) + \mcFG (X(t)).
\end{equation*}
\noindent By subtracting Milstein approximations of $X$ and $Y$, we get
\begin{align}\label{equ:difference}  
\begin{split}
X(t+h)-Y(t&+h) = X(t) -Y(t)\\
&+\mcFG (X(t))-\mcFG (Y(t)).
\end{split}
\end{align}
Taking $\|\cdot\|_{\mc X}$, raising to the power $l$,  taking expected value $\expt$,  subtracting $\expt\|X(t) -Y(t)\|_{\mc X}^l$ from both sides,  dividing by $h$, and taking limit as $h\to0^+$, we get (to fit the equations, we dropped some of $(t)$ arguments):
{\small{\begin{align*}
&\lim_{h\to0^+}\frac{1}{h}\left\{\expt\|X(t+h) -Y(t+h)\|_{\mc X}^l - \expt\|X(t) -Y(t)\|_{\mc X}^l\right\}\\
&=\lim_{h\to0^+}\frac{1}{h}\Big\{\expt\|X(t)-Y(t)+\mcFG(X) -\mcFG (Y)\|^l_{\mc X}\Big. \\
&\qquad\qquad\quad\Big.- \expt\|X(t) -Y(t)\|_{\mc X}^l\Big\}\\
&=\lim_{h\to0^+}\frac{1}{h}
\Big\{\frac{\expt\|X-Y+\mcFG(X) -\mcFG (Y)\|^l_{\mc X}}{ \expt\|X(t) -Y(t)\|_{\mc X}^l}-1\Big\} \\
&\qquad\qquad\qquad\times\expt\|X(t) -Y(t)\|_{\mc X}^l\\
&\leq \mc M^{+}_{\mc X,l}[F,G]\; \expt\|X(t) -Y(t)\|_{\mc X}^l, 
\end{align*}
}}
where the last inequality holds by the definition of $\mc M^{+}_{\mc X,l}[F,G]$ and the fact that $X(t)-Y(t)$ is non-anticipating, and  hence, independent of the Wiener increment $dW$. 
The first term of the above relationships is the upper Dini derivative of $\expt\|X(t) -Y(t)\|_{\mc X}^l$. Hence, 
\begin{align*}
D^+\expt\|X(t) -Y(t)\|_{\mc X}^l \leq \mc M^{+}_{\mc X,l}[F,G]\; \expt\|X(t) -Y(t)\|_{\mc X}^l. 
\end{align*}
 Applying comparison lemma \cite[Lemma 3.4]{HKK:02}, $\forall t\geq0$:
 \begin{align*}
\expt\|X(t) -Y(t)\|_{\mc X}^l \leq  \expt\|X(0) -Y(0)\|_{\mc X}^le^{\mc M^{+}_{\mc X,l}[F,G] t}, 
\end{align*}
which is the desired result. 
\end{proof}

  In this work, inspired by \textit{common} noise-induced synchronization, we assume that all the trajectories realize the \textit{same} Wiener process $W$ and therefore \eqref{equ:difference} is a valid equation.
Studying stochastic contractivity for the trajectories driven by distinct Wiener processes is a topic of future investigations. 

 Next proposition gives an upper bound for $\mc M^{+}_{\mc X,l}[F,G]$ based on the deterministic logarithmic  norms of $J_F$ and $J_{G_j}$, $j=1,\ldots,d$.  The upper bound  makes the result of Theorem \ref{thm:stochastic:contractivity} more applicable, since computing deterministic logarithmic  norms induced by some norms, such as $L^p$ norms and weighted $L^p$  norms for $p=1,2,\infty$ are straightforward. 

\begin{proposition}\textbf{\textit{\small(Relationship between deterministic and stochastic logarithmic Lipschitz constants)}}\label{prop:relation:stochM:M}
Let $F$, $G$, and $W$ be as described in Notation \ref{notation}. Then 
\begin{align}\label{eq1:relation:stochM:M}
\begin{split}
\mc M^{+}_{\mc X,l}[F,G] &\leq lM^+_{\mc X}\Big[F-\frac{1}{2}\sum_jJ_{G_j} G_j\Big]\\
&+\frac{l}{\sqrt{2\pi}}\sum_j (M^+_{\mc X}[G_j] +M^+_{\mc X}[-G_j]).
\end{split}
\end{align}
Furthermore, if $F$ and $G_j$s are continuously differentiable and $\mc Y$ is convex, then the following inequality holds. 
\begin{align}\label{eq2:relation:stochM:M}
\begin{split}
\mc M^{+}_{\mc X,l}[&F,G] \leq l\sup_x\mu_{\mc X}\Big[J_{F-\frac{1}{2}\sum_jJ_{G_j}G_j}(x)\Big]\\
&+ \frac{l}{\sqrt{2\pi}}\sum_j(\sup_x\mu_{\mc X}[J_{G_j}(x)] +\sup_x\mu_{\mc X}[-J_{G_j}(x)]). 
\end{split}
\end{align}
\end{proposition}

\noindent See Appendix \ref{sec:appendix} for a proof. 
 
\begin{corollary}\textbf{\textit{(Stochastic contraction based on deterministic logarithmic norms)}}\label{cor:contractivity:1}
 Under the conditions of Proposition \ref{prop:relation:stochM:M}, if there exists $c>0$ such that the right hand side of \eqref{eq2:relation:stochM:M} is upper bounded by $-c$, 
then \eqref{SDE:eq} is  $l-$th moment stochastically contractive. 
\end{corollary}

\begin{proof}
Since the right hand side of \eqref{eq2:relation:stochM:M} is bounded by $-c$, so is its left hand side, i.e., 
$\mc M^{+}_{\mc X,l}[F,G]\leq -c$. 
Therefore, by Theorem \ref{thm:stochastic:contractivity}, system \eqref{SDE:eq} is stochastically contractive.
\end{proof}

\section{Noise-induced contractivity and synchronization}\label{sec:noise-induced:contrac}

 In this section we show how a multiplicative noise can be beneficial for a system and make it contractive. 
Suppose $\dot x = F(x)$
 is not contractive, that is, for any given norm $\|\cdot\|_{\mc X}$, $\sup_{x}\mu_{\mc X}[J_F(x)]\geq 0$. Corollary \ref{cor:contractivity:1} suggests that for appropriate choices of the noise term $G$ and norm $\|\cdot\|_{\mc X}$, the underlying stochastic system $dx = F(x)dt +G(x)dW$ may become stochastically contractive. 
 The reason is that there might exist $G$ and $\|\cdot\|_{\mc X}$ such that for any $x$, 
 $\mu_{\mc X}[J_{F-\frac{1}{2}\sum_jJ_{G_j}G_j} (x)]$ becomes a small enough negative number.
Note that by sub-additivity of the logarithmic norms, 
$0=\mu_{\mc X}[J_{G_j}-J_{G_j}]\leq  \mu_{\mc X}[J_{G_j}]+\mu_{\mc X}[-J_{G_j}].$
Hence, the last sum on the right hand side of \eqref{eq2:relation:stochM:M} is always non-negative. Therefore, the first term must be small enough such that the sum becomes negative. 
 For example, for a linear diffusion term, i.e.,  $G_j(x) = \sigma_j x$, $\sigma_j>0$:
\[\mu_{\mc X}[J_{G_j}]+\mu_{\mc X}[-J_{G_j}] = \sigma (\mu_{\mc X}[I]+\mu_{\mc X}[-I]) =0,\] 
and by sub-additivity of logarithmic norms: 
\begin{align}\label{sub-additivity}
\begin{split}
\mu_{\mc X}\left[J_{F-\frac{1}{2}\sum_jJ_{G_j}{G_j}}\right] &= \mu_{\mc X}\Big[J_F-\frac{1}{2}\sum_j\sigma_j^2I\Big] \\
&\leq \mu_{\mc X}\left[J_F\right]-\frac{1}{2}\sum_j\sigma_j^2. 
\end{split}
\end{align}
 For some large $\sigma_j$s, $\mu_{\mc X}[J_F]-\frac{1}{2}\sum_j\sigma_j^2$ becomes negative, and hence, the SDE becomes stochastically contractive. 
Intuitively, since we assumed all the trajectories sense the same Wiener process, the noise plays the role of a common external force to all the trajectories. Therefore, for a strong enough noise, the trajectories converge to each other.
 See Example \ref{example:theorem1} below.  
 
 Equation~\eqref{sub-additivity}  guarantees that linear multiplicative stochastic terms do not destroy the contraction properties of contraction systems, no matter how large the perturbations are. 

  Note that  Corollary \ref{cor:contractivity:1}  argues that  multiplicative noise may aid contractivity. For an additive noise, i.e., for a state-independent noise term  $G_j(x)\equiv a$, $\mu_{\mc X}[J_{G_j}]=\mu_{\mc X}[-J_{G_j}]=0$ and $\mu_{\mc X}[J_{F-\frac{1}{2}J_{G_j}{G_j}}] = \mu_{\mc X}[J_F]$. Therefore, $\mc M^{+}_{\mc X,l}[F,G] \leq \mu_{\mc X}[J_F]$ and  $\mu_{\mc X}[J_F]\geq0$ do not give any information on the sign of $\mc M^{+}_{\mc X,l}$, and hence, on the contractivity of the SDE.  

\begin{example}\label{example:theorem1}
	We consider the Van der Pol oscillator subject to a multiplicative noise 
     \begin{align}\label{eq:vanderpol}
     \begin{split}
           dx &= \left(x - \frac{1}{3} x^3 -y \right)dt +\sigma g_1(x) dW,\\
           dy &= x dt +\sigma g_2(y) dW, 
    \end{split}
    \end{align}
    where we assume that the Wiener process is one dimensional, $d=1$.
  The state of the oscillator  is denoted by $X=(x,y)^\top$ which its change of rate is described by $F=(x - \frac{1}{3} x^3 -y , x)^\top$. The noise of the system is described by the column vector $G(x,y) = (\sigma g_1(x), \sigma g_2(y))^\top$. 

A simple calculation shows that the Jacobian of $F$ evaluated at the origin is not Hurwitz, i.e., the eigenvalues are not negative. Therefore, the deterministic Van der Pol is not contractive with respect to any norm. 
 Figure~\ref{fig:vandepol:no-noise} depicts two trajectories $(x_1,y_1)^\top$ and $(x_2,y_2)^\top$ of \eqref{eq:vanderpol} in the absence of noise which do not converge. 

In Figure~\ref{fig:vandepol:additive-noise}, an additive noise $g_1(x)=g_2(y)=1$ with intensity $\sigma=0.35$ is added. We observe that the trajectories still do not converge. As discussed above,  our result in Corollary \ref{cor:contractivity:1} does not guarantee noise-induced contractivity in the case of additive noise. 

In Figure~\ref{fig:vandepol:with-noise}, a state-dependent multiplicative noise 
 $(g_1(x),  g_2(y))=(1+4x, 1+4y)$ with noise intensity $\sigma=0.35$ is added and two trajectories with initial conditions $(1,-1)^\top$ and $(2,-2)^\top$  are plotted. We observe that the two trajectories converge to each other. A simple calculation shows that $\mu_2[J_G]+ \mu_2[-J_G] = 4\sigma -4\sigma=0$, where $\mu_2[A] = \frac{1}{2} \max\lambda(A+A^\top)$ is the logarithmic norm induced by $L^2$ norm and  $\max\lambda$ denotes the largest eigenvalue. Also, 
 \begin{align*}
 \sup_{(x,y)}\mu_2[J_{F-\frac{1}{2}J_GG}(x,y)] &= \sup_{(x,y)}\max\{1-x^2-8\sigma^2, -8\sigma^2\}\\
 &=1-8\sigma^2.
  \end{align*}
   By Corollary \ref{cor:contractivity:1}, 
 $\mc M^{+}_{\mc X,2}[F,G]\leq 2 (1-8\sigma^2)$. Therefore, for $\sigma>\frac{1}{\sqrt 8} \approx 0.35$, $\mc M^{+}_{\mc X,l}[F,G]<0$ and the system becomes  $l-$th moment stochastically contractive for any $l\geq1$.
 
 Figure \ref{fig:expected-value} shows the mean square difference of the two solutions plotted in Figure~\ref{fig:vandepol:with-noise} over 5000 simulations, which converges to zero, as expected. 
\end{example}
 		\begin{figure}
	\centering	               
		\subfigure[No contraction in deterministic system]{
			\includegraphics[width=0.36\textwidth]{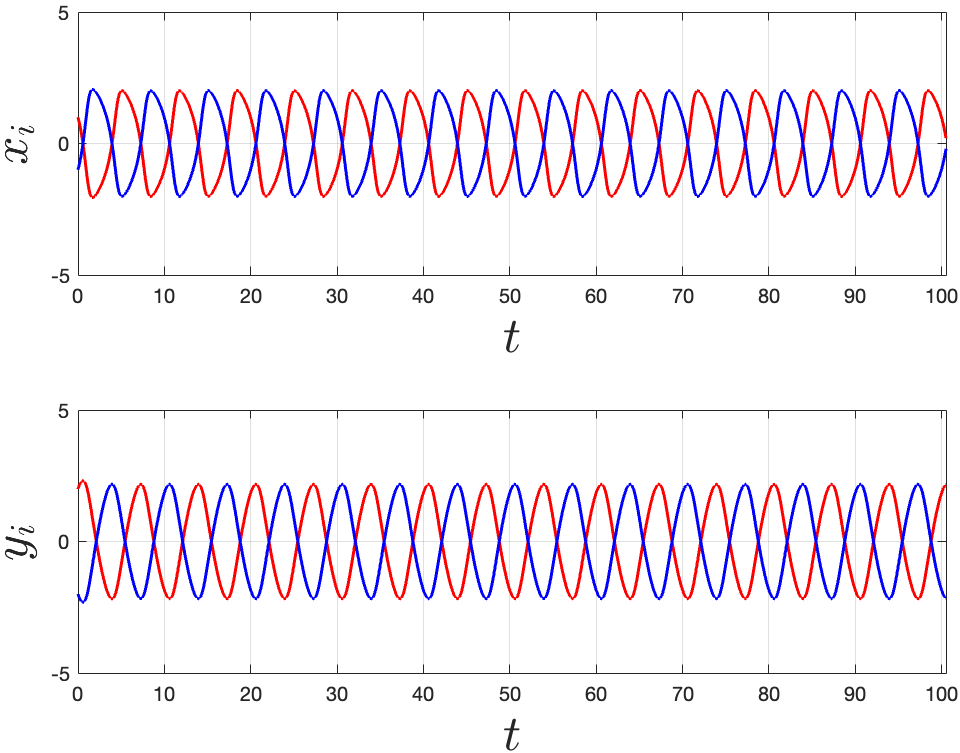}
			\label{fig:vandepol:no-noise}
		}
		
				\subfigure[No contraction with additive noise]{
			\includegraphics[width=0.36\textwidth]{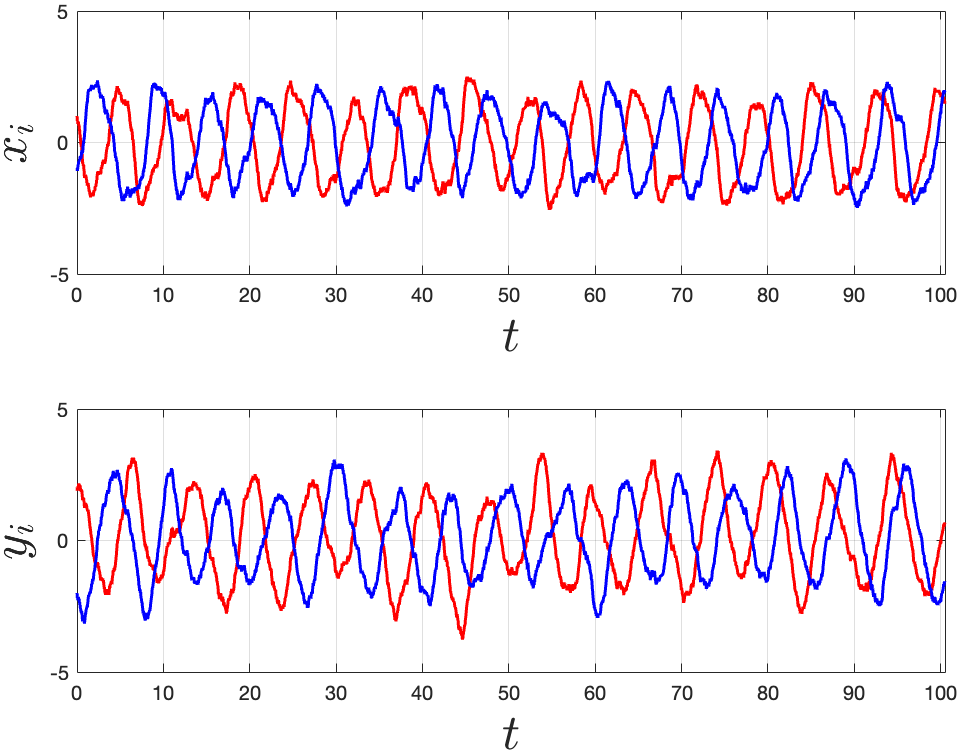}
			\label{fig:vandepol:additive-noise}
		}
			\subfigure[Multiplicative noise-induced contraction]{
		\includegraphics[width=0.36\textwidth]{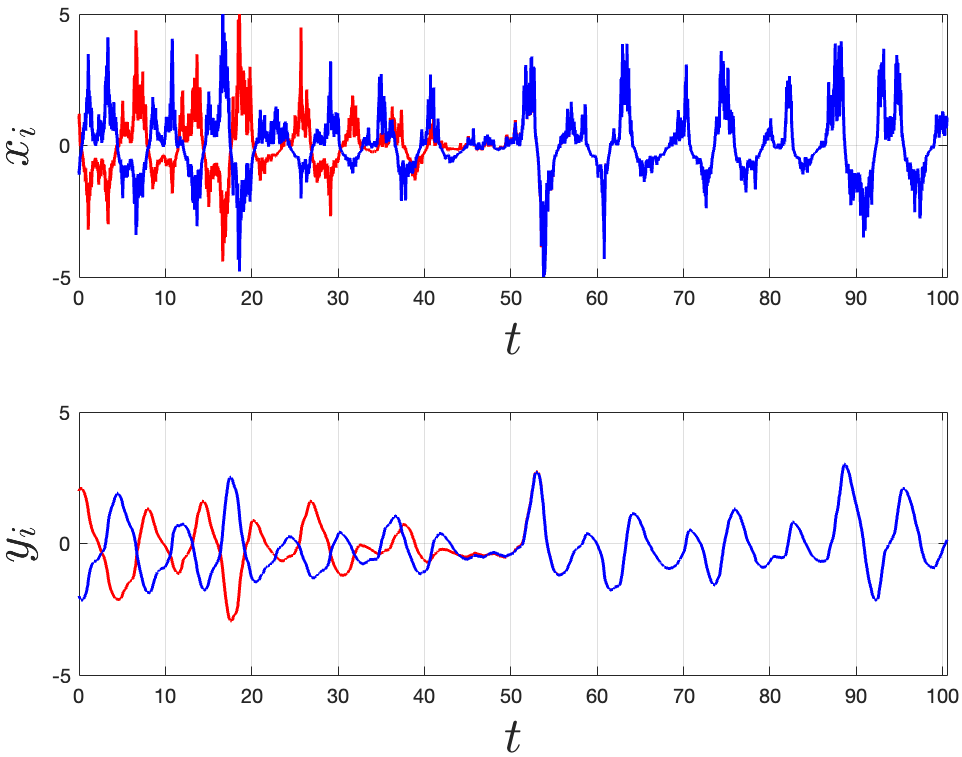}  
		\label{fig:vandepol:with-noise}
	}
	\subfigure[The mean square of difference of two solutions]{
	\includegraphics[width=0.39\textwidth]{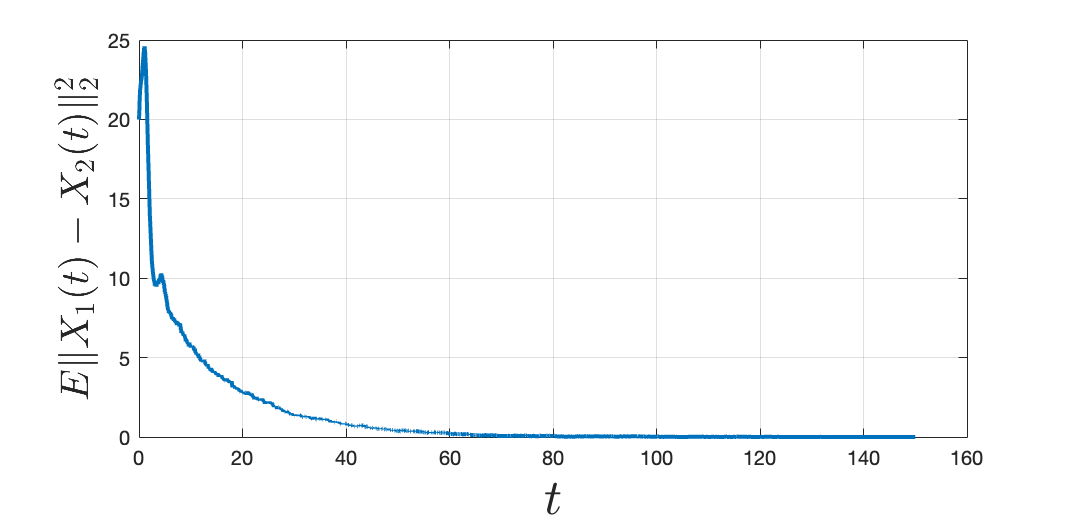}
	\label{fig:expected-value}
	}
		\caption{Contraction behavior of van der Pol oscillator given in Example \ref{example:theorem1}.  
		(a) Two trajectories of the deterministic oscillator are plotted to show the system is not contractive. 
		(b) An additive noise ($g_1(x)=g_2(y)=1$) with intensity $\sigma=0.35$ is added but does not make the system contractive. 
		 (c)  A multiplicative noise ($g_1(x)=1+4x, g_2(y)=1+4y$) with intensity $\sigma=0.35$ is added which makes the system contractive. 
		 (d) The mean square difference of two solutions over 5000 simulations is shown. 
	}
	\end{figure}

 Now consider a network of $N$ isolated nonlinear systems which are driven by a multiplicative common noise, i.e., the only interaction between the systems is through the common noise. The dynamics of such a network can be described by the following SDEs. For $i=1,\ldots, N,$
\begin{align}\label{eq:synchronization}
d X_i = F(X_i) dt +\sigma G(X_i) dW,
\end{align}
with initial conditions $X_i(0)=X_{i0}$. 
Then \eqref{eq:synchronization}  stochastically synchronizes if for any $i,j,$ $\expt\|X_i(t)-X_j(t)\|_{\mc X}^l\to0$ as $t\to\infty$, which can be concluded from  $d X = F(X) dt +\sigma G(X) dW$ being contractive. 

\section{Conclusions}\label{sec:conclusion}

 Deterministic logarithmic Lipschitz constants generalize classical logarithmic norms to nonlinear operators. These constants are proper tools to characterize the contraction properties of ODEs.   In this work, we introduced the notions of \textit{stochastic} logarithmic Lipschitz constants and used them to extend contraction theory to a class of SDEs. Unlike some logarithmic norms,  computing stochastic (or deterministic) logarithmic Lipschitz constants is not straightforward.   
 Therefore, to make our theory more applicable, we found some relationships between stochastic logarithmic Lipschitz constants and logarithmic norms. 
We discussed how multiplicative noise could aid contractivity and foster stochastic synchronization in nonlinear dynamical systems.  

In this paper, we assumed that a common Wiener process drives all the trajectories. Studying contractivity (respectively, network synchronization) in the case that distinct and independent Wiener processes drive the trajectories (respectively, nonlinear dynamical systems) is a topic of future investigations.  
In this case, we need to define an ``approximate" contraction in the sense that the trajectories exponentially enter a tube and stay there but do not necessarily converge. See \cite{slotine_stochastic} 
 (respectively, \cite{2021_Aminzare_Srivastava_Cybernetic}) for this type of contractivity (respectively, synchronization) which are based on stochastic Lyapunov function. 
Proposition \ref{prop:relation:stochM:M} provides a mechanism to characterize stochastic contractivity in a class of nonlinear SDEs and understand stochastic synchronization in networks driven by common noise. Generalizing this result to the case of independent Wiener increments is another topic of future investigations. 

\section*{ACKNOWLEDGMENT}
   The author would like to thank Michael Margaliot for his comments that improved this paper's exposition. This work is supported by Simon Foundations' grant $712522.$

\section{Appendix}\label{sec:appendix}

\noindent\textbf{Proof of Proposition~\ref{prop:properties:stochastic:constants}.}

\noindent{1.} For $h>0$, let $\Omega(h) = \frac{\|u-v+hF(u) -hF (v)\|_{\mc X}}{\|u-v\|_{\mc X}}$. 
Using the equality $\Omega^l-1= (\Omega-1)(\Omega^{l-1}+\cdots+1)$ and the fact that $\lim_{h\to0}\Omega(h)=1$, we have, 
{\small{
\begin{align*}\label{}
&\mc M^{+}_{\mc X,l}[F, 0]\\
&=\displaystyle\sup_{u\neq v\in \mc Y}\lim_{h\to0^{+}}\frac{1}{h}
\left(\frac{\|u-v+h(F(u) -F (v))\|^l_{\mc X}}{\|u-v\|^l_{\mc X}}-1\right)\\
&=\displaystyle\sup_{u\neq v\in \mc Y}\lim_{h\to0^{+}}\frac{l}{h}
\left(\frac{\|u-v+h(F(u) -F (v))\|_{\mc X}}{\|u-v\|_{\mc X}}-1\right) \\
&= l M_{\mc X}^{+}[F].
\end{align*}
}}
\noindent{2.} The proof is straightforward from the definition of $\mc M^{+}_{\mc X,l}$ and $\mc M_{\mc X,l}$. 

\noindent{3.} By the definition of $\mcFG$ given in Notation \ref{notation},
\[\mc M_{\alpha F,\sqrt\alpha G}^{(h,W)} = \mc M_{ F, G}^{(\alpha h,\sqrt\alpha W)}.\]
 Using the fact that $W$ is of order $\sqrt h$, and therefore, $\sqrt \alpha W$ is of order $\sqrt{\alpha h}$, we have:
{\small{\begin{align*}\label{}
&\mc M^{+}_{\mc X,l}[\alpha F,\sqrt\alpha G]=\displaystyle\sup_{u\neq v\in \mc Y}\lim_{h\to0^{+}}\frac{1}{h}\times\\
&\left(\expt\;\frac{\|u-v+\mc M_{ F, G}^{(\alpha h,\sqrt\alpha W)}(u) -\mc M_{ F, G}^{(\alpha h,\sqrt\alpha W)} (v)\|^l_{\mc X}}{\|u-v\|^l_{\mc X}}-1\right)\\
&=\displaystyle\sup_{u\neq v\in \mc Y}\lim_{h\to0^{+}}\frac{\alpha}{\alpha h}\times\\
&\left(\expt\;\frac{\|u-v+\mc M_{ F, G}^{(\alpha h,\sqrt\alpha W)}(u) -\mc M_{ F, G}^{(\alpha h,\sqrt\alpha W)} (v)\|^l_{\mc X}}{\|u-v\|^l_{\mc X}}-1\right)\\
&=\alpha \mc M_{\mc X, l}^{+}[ F, G].
\end{align*}
}}
\noindent The second inequality in part 3 can be obtained by the definition of $\mc M^{+}_{\mc X,l}$ and the following equality.
\[2\mc M_{F^1+F^2, G^1+G^2}^{(h,W)} = \mc M_{ F^1, \frac{G^1+G^2}{\sqrt 2}}^{(2 h,\sqrt 2 W)}+ \mc M_{ F^2, \frac{G^1+G^2}{\sqrt 2}}^{(2 h,\sqrt 2 W)}.\]
\oprocend 

\textbf{Proof of Proposition~\ref{prop:relation:stochM:M}.}
For fixed $u, v, F, G,$ and $h>0$, define $\mc K(h)$ and $\mc K_l(h)$ as follows: 
\begin{align*}
\mc K(h) &:= \frac{\|u-v+\mcfg(u) -\mcfg (v)\|_{\mc X}}{\|u-v\|_{\mc X}}, \\
\mc K_l(h) &:= \expt\;\frac{\|u-v+\mcfg(u) -\mcfg (v)\|^l_{\mc X}}{\|u-v\|^l_{\mc X}}\\
&=\expt\; \mc K(h)^l.
\end{align*}
Note that $\mc M^{+}_{\mc X,l}[F,G] = \sup_{u\neq v} D^+\mc K_l(0)$.  A simple calculation shows that the derivative of $\mc K_l$ evaluated at $h=0$ is equal to 
$
l  \;\expt(\mc K(h)^{l-1} \mc D^+\mc K(h))|_{h=0} = l\;\expt (D^+\mc K(0)), 
$
since $\mc K(0)=1$ and by Dominated Convergence Theorem, the limit in the definition of $D^+$ and the expected value  can be exchanged. 
Therefore, by the definition of upper Dini derivative: 
{\small{
\begin{align*}
&D^+\mc K(0)\\
&= \lim_{h\to 0^+} \frac{l}{h} \left\{\frac{\|u-v+\mcfg(u) -\mcfg (v)\|_{\mc X}}{\|u-v\|_{\mc X}} -1 \right\}\\
&\leq \lim_{h\to 0^+} \frac{l}{2h} \left\{\frac{\|u-v+\mc M_{F,G}^{(2h,0)}(u) -\mc M_{F,G}^{(2h,0)} (v)\|_{\mc X}}{\|u-v\|_{\mc X}} -1 \right\}\\
&+\lim_{h\to 0^+} \frac{l}{2h} \left\{\frac{\|u-v+\mc M_{F,G}^{(0,2W)}(u) -\mc M_{F,G}^{(0,2W)} (v)\|_{\mc X}}{\|u-v\|_{\mc X}} -1 \right\}\\
&\leq M^+_{\mc X} \Big[F-\frac{1}{2}\sum_jJ_{G_j} G_j\Big] \\
&+\lim_{h\to 0^+} \frac{l}{2h}\left\{\frac{\|u-v+ 2\sum_j(G_j(u)-G_j(v))\Delta W_j\|_{\mc X}}{\|u-v\|_{\mc X}} -1 \right\}.
\end{align*}
}}
 The first inequality is obtained  by writing   $\mcfg$ as $ \mc M_{F,G}^{(h,0)}+\mc M_{F,G}^{(0,W)}$. 
The second inequality is obtained using the following expressions: 
$\mc M_{F,G}^{(h,0)}= h\big(F - \frac{1}{2}\sum_jJ_{G_j} G_j\big)$ and 
$\mc M_{F,G}^{(0,W)}=\sum_j G_j\Delta W_j + \mathcal{R}$. 
The term $\mathcal{R}(u)- \mathcal{R}(v)$ is omitted from the last term because $\mathcal{R}$ contains a factor of $h^2$ which vanishes when $h\to0$. 
The Wiener increment  satisfies
\[\Delta W_j= W_j(t+h) - W_j(t) =\int_{t}^{t+h} \xi_j(s)ds =  h\xi_j(\tilde s_j),\]
 where $\xi_j$ is the standard normal distribution, $\xi_j\sim \mc{N}(0,1)$ and $t<\tilde s_j<t+h$. 
Since with probability $\frac{1}{2}$, $\xi_j$ is positive or negative, the last term of the above relationship becomes:
{\small{
\begin{align*}
&\lim_{h\to 0^+} \frac{1}{2h}\left\{\frac{\|u-v+ 2h\sum_j(G_j(u)-G_j(v))\xi_j(\tilde s_j)\|_{\mc X}}{\|u-v\|_{\mc X}} -1 \right\}\\
&\leq M^+_{\mc X} \Big[\sum_j\xi_j G_j\Big]\leq \sum_jM^+_{\mc X} [\xi_j G_j]\\\
&= \frac{1}{2}\sum_j\left(M^+_{\mc X} [|\xi_j|G_j]+  M^+_{\mc X} [-|\xi_j|G_j]\right)\\
&= \frac{1}{2} \sum_j  |\xi_j| \left(M^+_{\mc X} [G_j]+  M^+_{\mc X} [-G_j]\right).
\end{align*}
}}
Therefore, by taking $\expt$ from both sides, we get
\begin{align*}
l\;\expt \Big(D^+\mc K(0)\Big)&\leq l \sum_j M^+_{\mc X} \Big[F-\frac{1}{2}\sum_jJ_{G_j} G_j\Big] \\
&+ \frac{l}{2} \sum_j \expt |\xi_j| \left(M^+_{\mc X} [G_j]+  M^+_{\mc X} [-G_j]\right).
\end{align*}
Equation~\eqref{eq1:relation:stochM:M} is obtained by plugging $\expt |\xi_j| = \sqrt\frac{2}{\pi}$.  

\noindent Equation~\eqref{eq2:relation:stochM:M} holds by Proposition~\ref{prop:logarithmic:constants:norm}.
\oprocend



\begin{thebibliography}{99}

\bibitem{Lewis1949}
D.~C. Lewis, ``Metric properties of differential equations,'' {\em Amer. J.
  Math.}, vol.~71, pp.~294--312, 1949.

\bibitem{Hartman1961}
P.~Hartman, ``On stability in the large for systems of ordinary differential
  equations,'' {\em Canad. J. Math.}, vol.~13, pp.~480--492, 1961.

\bibitem{Soderlind}
G.~Soderlind, ``The logarithmic norm. history and modern theory,'' {\em BIT},
  vol.~46, no.~3, pp.~631--652, 2006.

\bibitem{dahlquist}
G.~Dahlquist, {\em Stability and error bounds in the numerical integration of
  ordinary differential equations}.
\newblock Inaugural dissertation, University of Stockholm, Almqvist \& Wiksells
  Boktryckeri AB, Uppsala, 1958.

\bibitem{Demidovich1961}
B.~P. Demidovi{\v{c}}, ``On the dissipativity of a certain non-linear system of
  differential equations. {I},'' {\em Vestnik Moskov. Univ. Ser. I Mat. Meh.},
  vol.~1961, no.~6, pp.~19--27, 1961.

\bibitem{Demidovich1967}
B.~P. Demidovi{\v{c}}, {\em Lektsii po matematicheskoi teorii ustoichivosti}.
\newblock Izdat. ``Nauka'', Moscow, 1967.

\bibitem{Yoshizawa1966}
T.~Yoshizawa, {\em Stability theory by {L}iapunov's second method}.
\newblock Publications of the Mathematical Society of Japan, No. 9, The
  Mathematical Society of Japan, Tokyo, 1966.

\bibitem{Yoshizawa1975}
T.~Yoshizawa, {\em Stability theory and the existence of periodic solutions and
  almost periodic solutions}.
\newblock Springer-Verlag, New York-Heidelberg, 1975.
\newblock Applied Mathematical Sciences, Vol. 14.

\bibitem{Soderlind2}
G.~Soderlind, ``Bounds on nonlinear operators in finite-dimensional {B}anach
  spaces,'' {\em Numer}, vol.~50, no.~1, pp.~27--44, 1986.

\bibitem{Loh_Slo_98}
W.~Lohmiller and J.~J.~E. Slotine, ``On contraction analysis for non-linear
  systems,'' {\em Automatica}, vol.~34, pp.~683--696, 1998.

\bibitem{arcak2011_Automatica}
M.~Arcak, ``Certifying spatially uniform behavior in reaction-diffusion {PDE}
  and compartmental {ODE} systems,'' {\em Automatica}, vol.~47, no.~6,
  pp.~1219--1229, 2011.

\bibitem{Simpson-Porco_Bullo_2014}
J.~W. Simpson-Porco and F.~Bullo, ``Contraction theory on {R}iemannian
  manifolds,'' {\em Systems Control Lett.}, vol.~65, pp.~74--80, 2014.

\bibitem{russo_dibernardo_sontag13}
G.~Russo, M.~di~Bernardo, and E.~Sontag, ``A contraction approach to the
  hierarchical analysis and design of networked systems,'' {\em IEEE
  Transactions Autom. Control}, vol.~58, pp.~1328--1331, 2013.

\bibitem{contraction_survey_CDC14}
Z.~Aminzare and E.~Sontag, ``Contraction methods for nonlinear systems: A brief
  introduction and some open problems,'' in {\em Proc. IEEE Conf. Decision and
  Control, Los Angeles}, pp.~3835--3847, 2014.

\bibitem{Coogan_Arcak_2015}
S.~Coogan and M.~Arcak, ``A compartmental model for traffic networks and its
  dynamical behavior,'' {\em IEEE Transactions on Automatic Control}, vol.~60,
  no.~10, pp.~2698--2703, 2015.

\bibitem{Margaliot_Sontag_Tuller}
M.~Margaliot, E.~D. Sontag, and T.~Tuller, ``Contraction after small
  transients,'' {\em Automatica}, vol.~67, pp.~178--184, 2016.

\bibitem{2016_Aminzare_Sontag}
Z.~Aminzare and E.~D. Sontag, ``Some remarks on spatial uniformity of solutions
  of reaction–diffusion pdes,'' {\em Nonlinear Analysis: Theory, Methods \&
  Applications}, vol.~147, pp.~125--144, 2016.

\bibitem{Coogan_2019}
S.~Coogan, ``A contractive approach to separable {l}yapunov functions for
  monotone systems,'' {\em Automatica}, vol.~106, pp.~349--357, 2019.

\bibitem{CV_J_Bullo_2020}
P.~Cisneros-Velarde, S.~Jafarpour, and F.~Bullo, ``Contraction theory for
  dynamical systems on {Hilbert} spaces,'' Oct. 2020.

\bibitem{D_J_Bullo_2021}
A.~Davydov, S.~Jafarpour, and F.~Bullo, ``{Non-Euclidean} contraction theory
  for robust nonlinear stability,'' July 2021.
\newblock Submitted.

\bibitem{slotine_stochastic}
Q.~C. Pham, N.~Tabareau, and J.~Slotine, ``A contraction theory approach to
  stochastic incremental stability,'' {\em IEEE Transactions on Automatic
  Control}, vol.~54, no.~4, pp.~816--820, 2009.
    
  \bibitem{Han_Chung_2021}
S.~Han, and S.-J. Chung, ``Incremental nonlinear stability analysis for stochastic systems perturbed by Lévy noise"
  {\em arXiv:2103.13338}, 2021.

\bibitem{2013_Tabareau_Slotine}
N.~Tabareau and J.~Slotine, ``Contraction analysis of nonlinear random dynamical systems,"{\em arXiv:1309.5317v2}, 2013.

\bibitem{LG-TM-MM-2021}
L.~Gruene, T.~Kriecherbauer, and M.~Margaliot, ``Random attraction in the {tasep}
  model,'' {\em SIAM Journal on Applied Dynamical Systems}, vol.~20, no.~1,
  pp.~65--93, 2021.

\bibitem{newman_2018}
J.~Newman, ``Necessary and sufficient conditions for stable synchronization in
  random dynamical systems,'' {\em Ergodic Theory and Dynamical Systems},
  vol.~38, no.~5, p.~1857–1875, 2018.

\bibitem{2013_Pham_Slotine}
Q.~C. Pham and J.~Slotine, ``Stochastic contraction in {R}iemannian metrics,''
  {\em arXiv:1304.0340}, 2013.

\bibitem{2019_Bouvrieand_Slotine}
J.~Bouvrieand and J.~Slotine, ``Wasserstein contraction of stochastic nonlinear
  systems,'' {\em arXiv:1902.08567v2}, 2019.

  \bibitem{Dani_Chung_Hutchinson_2015}
A.~Dani, P.~Ashwin, S.J.~ Chung, and S.~Hutchinson, ``Observer Design for Stochastic Nonlinear Systems via Contraction-Based Incremental Stability"
{\em IEEE Transactions on Automatic Control},   vol.~60, no.~3, p.~1700--714, 2015.
  
\bibitem{Ahmad_Raha_2010}
S.~S. Ahmad and S.~Raha, ``On estimation of transient stochastic stability of
  linear systems,'' {\em Stochastics and Dynamics}, vol.~10, no.~03,
  pp.~385--405, 2010.

\bibitem{aminzare-sontag}
Z.~Aminzare and E.~D. Sontag, ``Logarithmic {L}ipschitz norms and
  diffusion-induced instability,'' {\em Nonlinear Analysis: Theory, Methods \&
  Applications}, vol.~83, pp.~31--49, 2013.

\bibitem{CG:09}
C.~Gardiner, {\em Stochastic Methods: A Handbook for the Natural and Social
  Sciences}.
\newblock Springer, fourth~ed., 2009.

\bibitem{HKK:02}
H.~K. Khalil, {\em Nonlinear systems}.
\newblock Prentice Hall, third~ed., 2002.

\bibitem{2021_Aminzare_Srivastava_Cybernetic}
Z.~{Aminzare} and V.~{Srivastava}, ``Stochastic synchronization in nonlinear
  network systems driven by intrinsic and coupling noise,'' {\em under review},
  2021.

\end{thebibliography}
\end{document}